\documentclass[10pt,conference]{IEEEtran}
\IEEEoverridecommandlockouts
\usepackage{cite}
\usepackage{amsmath,amssymb,amsfonts}
\usepackage[linesnumbered,ruled,vlined]{algorithm2e}
\SetAlFnt{\small}
\usepackage{graphicx}
\usepackage{geometry}
\usepackage{multirow}
\usepackage{textcomp}
\usepackage{xcolor}
\usepackage{comment}
\usepackage{amsthm}
\usepackage{adjustbox}
\usepackage{subfigure}
\usepackage{fancyhdr}
\newtheorem{theorem}{Theorem}
\newtheorem{lemma}[theorem]{Lemma}

\theoremstyle{definition}
\def\BibTeX{{\rm B\kern-.05em{\sc i\kern-.025em b}\kern-.08em
    T\kern-.1667em\lower.7ex\hbox{E}\kern-.125emX}}

\title{Fault-tolerant Network Design for Bounded Delay Data Transfer from PMUs to Control Center}
\author{asen }

\author{\IEEEauthorblockN{A. Sen, C. Sumnicht, S. Adeniye, D. Patel, S. Choudhuri}
\IEEEauthorblockA{\textit{School of Computing, Informatics and Decision Systems Engg.} \\
\textit{Arizona State University}\\
Tempe, AZ 85287, USA \\
\{asen, sadeniye, schoud13\}@asu.edu}}

\date{August 2020}

\begin{document}

\maketitle
\begin{abstract}

Communication network design for monitoring the state of an electric power grid has received significant attention in recent years.
In order to measure 
stability of a power grid, it is imperative that measurement data collected by the Phasor Measurement Units (PMUs) located at the Sub-Stations (SS) must arrive at the Control Center (CC) within a specified delay threshold $\delta$. In earlier papers we formalized the design problem as the \textit{Rooted Delay Constrained Minimum Spanning Tree (RDCMST)} problem. However, RDCMST does not provide any fault-tolerance capability, as failure of just one communication link would prevent PMU data from reaching the CC. In this paper, we
study the optimal cost network design problem with fault tolerance capability. In our model the PMU data from the SSs will have a path to reach the CC in spite of the failure of at most $R$ links 
within the 
delay threshold $\delta$. If $R = 1$, each SS will have two link disjoint paths of length at most $\delta$ to the CC. In other words, each SS will be on a cycle with the CC.
We refer to this problem as the \textit{Rooted Delay Constrained Minimum Cost Cycle Cover (RDCMCCC)} problem. We provide computational complexity analysis, an Integer Linear Programming formulation to find the optimal solution and a heuristic based on the sweeping technique. We evaluate the performance of our heuristic with real substation location data of Arizona.
\end{abstract}

\section{Introduction}
Accurate estimation of the operational state of the electric power grid is vital for its stable operation. For state estimation of the grid, power transmission control data, generated by the Phasor Measurement Units (PMUs) located in some of the Substations (SS), has to arrive at the  
Controls Center(s) (CC) within an acceptable delay threshold $\delta$. Accordingly, communication network topology design for 
electric power grid has received significant attention in recent times
\cite{ye2020multiple,zhu2018optimal,mohammadi2015new}.
    The authors in \cite{ye2020multiple} study the communication network design problem for the power grid to 
    analyze possible topologies,  
    and provide 
    guidance to power utilities for choosing the most appropriate communication network for their PMU based real-time applications. 
    
Many researchers have studied the optimal PMU placement problem with specific objectives, such as, full network observability \cite{baldwin1993power}.
However, optimal PMU placement without taking into account the cost of Communication Infrastructure (CI), may lead to the design of a very expensive Wide Area Monitoring (WAM) system, where the cost of CI may dominate the cost of the PMUs. Mohammadi {\em et.al.} \cite{mohammadi2015new} 
studied the
joint optimization of PMU placement and associated CI cost. This line of research is continued in \cite{zhu2018optimal}, where the authors study the optimal
PMU-communication link placement (OPLP) problem that simultaneously considers the placement of PMUs and communication links for
full observability.
    
    
Motivated by the studies undertaken by the power engineers \cite{ye2020multiple,zhu2018optimal,mohammadi2015new}, our earlier papers  formalized the design problem 
as the {\em Rooted Delay Constrained Minimum Spanning Tree (RDCMST) problem} \cite{sen2021optimal,sen2023delay}. However, this design does not provide any fault-tolerance capability: failure of just one communication link will prevent the PMU data from some of the substations from reaching the CC. In this paper, we overcome this limitation by studying the optimal cost network design problem with {\em fault tolerance capability}. That is, PMU data from the substations will have at least one path to reach the CC in spite of the failure of at most $R$ links, within 
the acceptable {\em delay threshold}. 
Although our formulation allows for the design of a $R$-fault-tolerant network for any $R$, our primary focus in this paper is the case where $R = 1$.
When $R = 1$, the problem may be viewed as a variation of the well-studied {\em cycle cover} and {\em vehicle routing} problems \cite{Blser2001ComputingCC, nagarajan2012approximation}. We define the problem as a {\em Rooted Delay Constrained Minimum Cost Cycle Cover (RDCMCCC) problem}.

The contributions of this paper are as follows:

\begin{itemize}

\item We introduce the delay bounded optimal cost fault-tolerant PMU to CC data transfer network design problem.
\item We prove that the RDCMCCC problem is NP-complete and provide an approximation algorithm for it.
\item We provide an Integer Linear Program for finding the optimal solution for the RDCMCCC problem.
\item We show that a variation of the RDCMCCC problem can be solved in polynomial time.
\item We provide a Sweeping Technique based Heuristic solution for the RDCMCCC problem.
\item We conduct extensive experimental evaluation of our heuristic with real Arizona substation data.


\end{itemize}


\section{Problem Formulation}
\label{sec2}

We formalize the RDCMCCC problem by viewing the locations of the substations as a set $P$ of $n$ points in the Euclidean plane. The input to the problems includes 
a distinguished point, $p_0 \in P$, corresponding to the location of the CC, and an acceptable delay threshold parameter $\delta$. We construct a {\em weighted, complete} graph $G = (V, E)$, where each node $v_i \in V$ corresponds to a point $p_i \in P$. Since the node $v_i$ and the point $p_i$ have a {\em one-to-one} correspondence, we use the terms {\em node} and {\em point} interchangeably. The weight $w(v_i, v_j)$ of an edge $(v_i, v_j) \in E$, is the {\em Euclidean distance} between the corresponding points $p_i, p_j \in P$. We assume that $w(v_i, v_j)$ represents both the {\em cost} and the {\em delay} of the link $(v_i, v_j)$. We note that this is a simplified assumption, as in a realistic communication system, the {\em link length} is only one of the factors that will determine the cost and the delay associated with that link. The cost of a link may be determined by factors such as wired/wireless, optical fiber/coaxial cable, etc, whereas the delay on a link may be determined by factors such as transmission rate, link bandwidth, queuing delay, etc. The consideration of only the link length (which determines the propagation delay) as the only factor makes the model somewhat simplistic. However, this simplified model is widely used in literature for the evaluation of communications networks \cite{xue2003minimum}.

\vspace{0.1in}
\noindent
{\em Rooted Delay Constrained Minimum Cost $(R = 1)$ Fault Tolerant Network Design Problem (RDCMC-(R=1)-FTNDP)}:

\vspace{0.05in}
\noindent
Given a directed complete graph $G = (V, E), V = \{v_0, \ldots, v_n\}, E = \{e_{i, j}, (v_i \rightarrow v_j), 1 \leq i, j \leq n, i \neq j\}$, a cost value $C_{i, j}$ and a delay value $D(i, j)$ associated with each directed edge $e_{i, j}$ , a distinguished vertex $v_0 \in V$ (referred to as the {\em root} node), and a delay bound parameter $\delta$. The goal of the RDCMC-(R=1)-FTNDP) problem is to find the {\em least cost subgraph} $G' = (V, E')$ 
of $G = (V, E)$, such that every node $v_i \in V, i \neq 0$ has at least two edge disjoint paths to $v_0$, and the length of each path 
is at most $\delta$. The {\em cost} of the subgraph $G'$ is equal to the sum of the cost of all the edges in $E'$. Similarly, the delay of the path $P_{v_i, v_0}, \forall i, 1 \leq i \leq n$ from node $v_i$ to $v_0$ is equal to the sum of the delays of all the edges  $e_{i, j} \in P_{v_i, v_0}$.

Note that in RDCMC-(R = 1)-FTNDP 
the directed edge $e_{i, j}, (v_i \rightarrow v_j)$ may be a part of the several paths, say, from $v_a$ to $v_0$, $v_b$ to $v_0$, $v_c$ to $v_0$ etc. 
In this environment, as the cost $C_{i, j}$ is associated with each edge $e_{i, j}$. 
if the edge  $e_{i, j}$ is a part of three paths, $v_a$ to $v_0$,  $v_b$ to $v_0$, $v_c$ to $v_0$, unless one is careful, the cost $C_{i, j}$ may be counted three times. However, in our model we view that this cost should not be counted three times, as an optical fiber line connecting the nodes $v_i$ and $v_j$ will be able to carry network traffic for multiple paths. Moreover, in our model if both the directed edges $e_{i, j}, (v_i \rightarrow v_j)$ and $e_{j, i}, (v_j \rightarrow v_i)$ are used for path construction, the two costs $C_{i, j}$ and $C_{j, i}$, associated with edges $e_{i, j}$ and $e_{j, i}$ respectively, should not be counted as two separate costs. In such a situation the cost of the link connecting the nodes $v_i$ and $v_j$ should be counted only one time, for the same reason given earlier. 

In the mathematical domain, it implies that once the solution of the RDCMC-R-FTNDP is found, where all the edges are directed, the orientation of the edges may be ignored and 
the resulting graph be treated 
as an undirected graph. Since in this undirected graph version of the solution, each node $v_i, 1 \leq i \leq n$, will have two undirected paths to the node $v_0$ of length at most $\delta$. In other words, each node $v_i, 1 \leq i \leq n$ will be on at least one cycle with the node $v_0$ where cycle length is at most $2\delta$.


From this perspective the RDCMC-(R = 1)-FTNDP is similar to the well-studied 
{\em Cycle Cover} and {\em Vehicle Routing} problems \cite{Blser2001ComputingCC, nagarajan2012approximation}. However, important differences remain, primarily in the way the cost of the design is measured. To the best of our knowledge this version has not been studied either in Cycle Cover or Vehicle Routing domain. Because of the similarity of the undirected version solution of the RDCMC-(F = 1)-FTNDP, we refer to it as the 
{\em Rooted Delay Constrained Minimum Cost Cycle Cover (RDCMCCC) Problem}.

\vspace{0.05in}
\noindent
{\em Instance:} A complete undirected graph $G = (V, E)$, with  $V = \{v_0, \ldots, v_n\}$ and a cost $C_{i, j}$ and delay $D_{i, j}$ associated with each edge $e_{i, j} \in E$, a distinguished node $v_0$, a delay threshold parameter $\delta$ and a cost threshold parameter $\gamma$.\\
{\em Question:} Is there a subgraph $G'$ of $G$ with the following properties: (i) every node $v_i, 1 \leq i \leq n$ is on a cycle with the node $v_0$, (ii) the total delay on each cycle is at most $\delta$ (the total delay on a cycle is equal to the sum of the delays associated with each of the edges constituting the cycle), and (iii) total cost of the edges of $G'$ is at most $\gamma$?

\section{Computational Complexity of RDCMCCC}
\label{CC_RDCMCCC}

\begin{figure*}[htp]
\centering
\includegraphics[width=.25\textwidth]{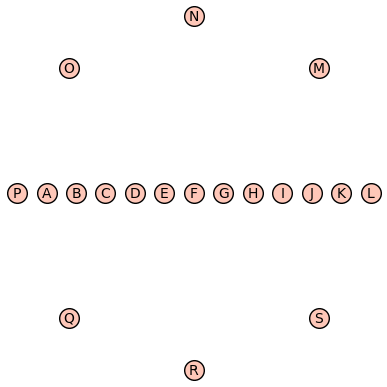}\hfill
\includegraphics[width=.25\textwidth]{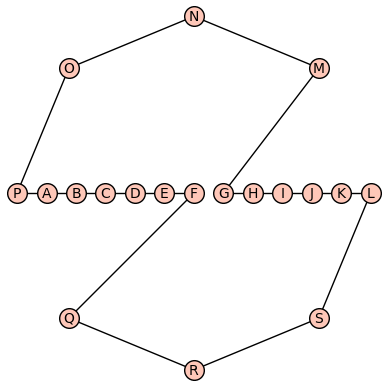}\hfill
\includegraphics[width=.25\textwidth]{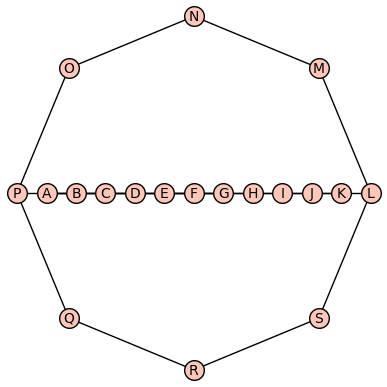}
\caption{(i) Set of points in the plane, (ii) TSP (One Cycle) Solution, (iii) Two Cycle Solution}
\label{fig:figure3}
\end{figure*}

\begin{theorem}
    RDCMCCC Problem is NP-complete.
\end{theorem}

We illustrate this by giving a transformation from the Hamiltonian Cycle (HC) problem \cite{garey1997computers}. 
We first establish the following lemma.  A collection of cycles $\mathcal{C} = \{C_1, C_2, \ldots C_m\}$ with $|\mathcal{C}| > 1$ is said to ``cover'' a graph $G$ if every node in $G$ is part of at least one cycle in $\mathcal{C}$.

\begin{lemma}\label{cyclelem}
    Consider a collection of cycles $\mathcal{C} = \{C_1, C_2, \ldots C_m, m > 1\}$ that covers a connected graph $G = (V, E)$ and each $C_i \in \mathcal{C}$ includes the vertex $v_1 \in V$. 
    This implies $V(C_i) \cap V(C_j) \neq \varnothing$ for all $i \neq j$.  If $E\left(\bigcup_{C\in\mathcal{C}}C\right)$ denotes the union of all the edges in the cycles belonging to $\mathcal{C}$, then $|E\left(\bigcup_{C\in\mathcal{C}}C\right)| > |V|$.
\end{lemma}

\begin{proof}
    As every vertex belongs to a cycle, every vertex must have a degree of at least two. Furthermore, as none of the cycles are disjoint, there must be at least one vertex with a degree strictly greater than $2$. Hence, by the degree sum formula
    \begin{align*}
        \sum_{v \in V} d(v) &\geq 2(|V|-1) + 3 \\
                            &= 2|V| + 1 \\
        \sum_{v \in V} d(v) &\geq 2|V| + 1
    \end{align*}
    On the other hand,
    \begin{align*}
        \sum_{v \in V} d(v) = 2\left|E\left(\bigcup_{C\in\mathcal{C}}C\right)\right|
    \end{align*}
    Therefore,
    \begin{align*}
        2\left|E\left(\bigcup_{C\in\mathcal{C}}C\right)\right| &\geq 2|V| + 1 \\
        \left|E\left(\bigcup_{C\in\mathcal{C}}C\right)\right| &\geq |V| + 1/2
    \end{align*}
    So, in fact,
    \begin{align*}
        \left|E\left(\bigcup_{C\in\mathcal{C}}C\right)\right| &\geq |V| + 1
    \end{align*}
\end{proof}

Observe $\left|E\left(\bigcup_{C\in\mathcal{C}}C\right)\right| = |V|$ happens when (i) $m = 1$, i.e., when the graph has a Hamiltonian Cycle or the cycles are all disjoint.

\noindent
{\em Proof of Theorem 1}: We restrict our attention to those instances of the RDCMCC problem for which the delay parameter $\delta$ is large, and as such, all solutions will satisfy the delay requirement. Thus, the only parameter that we need to pay attention to is the cost parameter $\gamma$.

Clearly, RDCMCC is in NP, as a Turing Machine can easily verify whether the cost of the subgraph $G'$ given as the output of the RDCMCC problem covers all the nodes and the cost of the edges in $G'$ doesn't exceed $\gamma$ \cite{garey1997computers}.

From an instance of the HC problem, we create an instance of the RDCMCC problem and establish that the input graph of the HC problem has a Hamiltonian Cycle if and only if the instance of the RDCMCC has a subgraph that covers all the nodes at a cost that doesn't exceed $\gamma$. 

Consider an instance $G = (V, E)$ of the HC problem. Construct an instance of the RDCMCC problem as follows: (i) create a complete edge-weighted undirected graph $G' = (V', E')$, where $V' = V$, (ii) if two nodes $v_i$ and $v_j$ have an edge between them in $G$, set $w(v_i, v_j) = 1$ in $G'$, (iii) if two nodes $v_i$ and $v_j$ do not have an edge between them in $G$, set $w(v_i, v_j) = 2$ in $G'$, (iv) set the cost parameter $\gamma = |V|$.

If $G$ has a Hamiltonian Cycle, clearly $G'$ has a subgraph that covers all the nodes in $G'$ in one cycle (including the node $v_1$), and the cost of this cycle is $\gamma = n$. Furthermore, if $G'$ has a subgraph that covers all the nodes in $G'$ using cycle(s) (each one which includes the node $v_1$) and the cost of this cycle is $\gamma = n$, then it must be done with only one cycle since if multiple cycles were used, we know from Lemma \ref{cyclelem} that the number of edges belonging to these cycles would have exceeded $|V| = n$ and would have exceeded the cost as the parameter $\gamma$ is set equal to $n$.

\subsection{Approximation Algorithm for the RDCMCCC Problem}

In this subsection we provide an approximation algorithm for the RDCMCCC problem arising out of the Geometric Space. By Geometric Space, we mean that each node $v_i$ of the RDCMCCC problem corresponds to a point $p_i$ in a metric space with distance function $d(\cdot, \cdot)$, and $C_{ij} = d(p_i, p_j)$.
We utilize the Christofides' approximation algorithm designed for the Traveling Salesman Problem (TSP) \cite{garey1997computers} to design the approximation algorithm for the RDCMCCC Problem in the geometric setting. 
Christofides' approximation algorithm 
provides a 1.5-factor guarantee in comparison with the optimal solution of the TSP \cite{garey1997computers}. 
We use the following notations to denote the optimal and approximate solutions to various problems on a graph $G$ with weights/costs associated with each edge.

\noindent
$MST$: Optimal Solution for the Minimum Spanning Tree problem for graph $G = (V, E)$\\
\noindent
$Chris~\_APP_{TSP}$: Christofides's Approximate Solution for the TSP problem for graph $G = (V, E)$\\
\noindent
$OPT_{RDCMCCC}$: Optimal Solution for the RDCMCCC problem for graph $G = (V, E)$\\
\noindent
$APP_{RDCMCCC}$: Approximate Solution for the RDCMCCC problem for graph $G = (V, E)$

The following inequalities hold between the above terms:
\begin{equation}
OPT_{RDCMCCC} > MST
\end{equation} 
Christofides' approximation algorithm for the TSP produces a solution that is at most 1.5 times the cost of the MST \cite{garey1997computers}. This approximation algorithm for the TSP can be used as an approximation algorithm for the RDCMCCC problem. 
\begin{equation}
APP_{RDCMCCC} = Chris~\_APP_{TSP} \leq 1.5 \times MST 
\end{equation}
From equations (1) and (2), we can conclude 

\vspace{-0.05in}


\begin{equation}
APP_{RDCMCCC} \leq 1.5 \times OPT_{RDCMCCC} 
\end{equation}


We make the following interesting observation regarding the RDCMCCC problem.

\vspace{0.1in}
\noindent
{\em Observation 1:} A multi-cycle solution may provide a lower cost design than a single-cycle solution.

Consider the following arrangement of points in Figure \ref{fig:figure3}.
Formally, the points corresponding to the vertices of the graph are described by $V = S_1 \cup S_2$ where:
\begin{align*}
    S_1 &:= \{(x/6, 0) : x \in [-5, 6]\} \\
    S_2 &:= \{(\cos(2\pi\theta/8), \sin(2\pi\theta/8)) : \theta \in [0, 8] \}
\end{align*}

Also, define the edge set to be $E = \{ (v_1, v_2, d(v_1, v_2)) : v_1, v_2 \in V\}$ where $d(\cdot, \cdot)$ is the standard Euclidean distance. The optimal Traveling Salesman tour (i.e., a 1-cycle solution to the RDCMCCC problem) may be computed with SageMath \cite{SageMath} 
which provides a solution with a tour length of $8.32$. However, a 2-cycle solution to the same set of input points is bounded above by $2r + 2\pi r = 2 + 2\pi \approx 8.28$. 

\section{Integer Linear Program for RDCMCCC}
\label{ILPSection}

We model the power grid as a graph $G = (\{v_0\} \cup V, E)$ where the node $v_0$ corresponds to the CC, the nodes $V = \{v_1, \ldots, v_n\}$ correspond to the SSs, and the edges in $E$ correspond to communication links (e.g., fiber optic lines). Let $|E| = m$. We view the physical location of the CC, as well as all the substations, as points in a two-dimensional plane. In the graph $G$, each node corresponds to a point in the plane, and the weight $w_{i,j}$ on the edge $e_{i,j}$, is taken to be the Euclidean distance between the corresponding points on the plane. The weight $w_{i,j}$ is taken to be a measure of both (i) the cost $C_{i,j}$  of the edge $e_{i,j}$ and also (ii) the delay $D_{i,j}$ on the edge $d_{i,j}$. We assume that the PMUs are placed only in a subset of the substations and denote this subset as $U \subseteq V$, $U = \{u_1, \ldots, u_s\}$. In some grid environments, $U$ may be equal to $V$, in which case $s = n$ and $\forall i, 1 \leq i \leq n, u_i = v_i$. As PMU data has to flow from the PMUs at the substations (nodes $u_1$ through $u_s$) to the CC (node $v_0$) within the specified delay threshold $\delta$, there will be $s$ flows (or paths) from the nodes $u_i, 1 \leq i \leq s$ to the node $v_0$, such that the length of each one of these paths must be upper bounded by $\delta$. 




The following points should be noted:

\begin{itemize}

\item In order to have a $R$-fault-tolerant communication network, each node $u_i, 1 \leq i \leq s$ must have $R + 1$ {\em edge disjoint} flows (or paths) from $u_i$ to $v_0$. Thus, the total number of paths will be $s \times (R + 1)$.
\item We use the notation $F^{k, r}$ to indicate the $r-$th flow (or path) from node $u_k$ to node $v_0$, $1 \leq k \leq s, 1 \leq r \leq R+1$.
\item  We use a binary variable $f^{k, r}_{i,j}$ defined as follows:
\[f^{k, r}_{i,j} = \left\{ \begin{array}{ll}
               1, & \mbox{if the flow $F^{k, r}$ uses the directed edge $e_{i, j}$} \\
               0, & \mbox{otherwise}
                   \end{array}
           \right. \]

\end{itemize}

\noindent
 In order to establish a flow (or path) from node $u_i, 1 \leq i \leq s$ to node $v_0$, the classical flow constraints must be imposed in the following way:\\

\noindent
Constraint at node $v_0$: for flow $f^{k, r}$, i.e., the $r$-th path from $u_k$ to $v_0$ ($v_0$ is the destination node),\\ $\forall k, 1 \leq k \leq s, \forall r, 1 \leq r \leq R+1$
 \begin{equation}
\sum_{j \in N(v_0)} f^{k, r}_{j, 0}-\sum_{j \in N(v_0)} f^{k, r}_{0, j} = 1
 \end{equation}

\vspace{0.1in}
\noindent
Constraint at node $u_k$, $1 \leq k \leq s$, $\forall r, 1 \leq r \leq R+1$ for flow $f^{k, r}$ ($u_k$ is the source node)
 \begin{equation}
\sum_{j \in N(u_k)} f^{k, r}_{j, k}-\sum_{j \in N(u_k)} f^{k, r}_{k, j} = -1
 \end{equation}

 \vspace{0.1in}
\noindent
Constraint at node $u_l, u_l \in \{V - \{v_0, u_k\}\}$ ($u_l$ is an intermediate node),  $\forall k, 1 \leq k \leq s, \forall r, 1 \leq r \leq R+1$

 \begin{equation}
   \sum_{j \in N(u_l)} f^{k, r}_{j, l}-\sum_{j \in N(u_l)} f^{k, r}_{l, j} = 0
 \end{equation}

 

            
                  
 \noindent
       Since all $R + 1$ paths from $u_k, 1 \leq k \leq s$ to $v_0$, $F^{k, r}, 1 \leq r \leq R+1$ must be edge-disjoint (i.e., no two paths should be using the same edge $e_{i, j}$ to satisfy the fault-tolerance requirement), we will have the following {\em Disjointness Constraint}.  $\forall k, 1 \leq k \leq s$
       \begin{equation}
           \sum_{r = 1}^{R + 1} 
       f^{k, r}_{i, j} \leq 1
       \end{equation}

\subsection{Computational Complexity of a variation of the RDCMCCC Problem}

       An Integer Linear Program with the following objective function 
       
    \begin{equation}
        {\tt ~Minimize} \sum_{i = 1}^n \sum_{j = 1}^n \sum_{k = 1}^s \sum_{r = 1}^{R +1} C_{i, j}f^{k, r}_{i, j}
       \end{equation}

       \noindent
       subject to the constraints (4), (5), (6) and (7) specified earlier may be viewed as a variation of the  RDCMCCC problem.  The version differs from the RDCMCCC in the following two ways:\\
       (i) In the RDCMCC problem, there is a constraint on the tour length of the cycle, in the sense that cycle lengths were not allowed to exceed the delay threshold value $\delta$. In this version there is no such constraint on the length of the cycle.\\
       (ii) In this version the cost of an edge between nodes $v_i$ and $v_j$ may be counted multiple times. First, if one flow from $v_a$ to $v_0$ uses the directed edge $v_i \rightarrow v_j$ and another flow from $v_b$ to $v_0$ uses the directed edge $v_j \rightarrow v_i$, the cost of the link connecting the nodes $v_i$ and $v_j$, $L_{i, j}$ (may be an optical fiber line) will be counted twice. Second, if multiple flows  $v_a$ to $v_0$, $v_b$ to $v_0$, $v_c$ to $v_0$ etc. uses the same directed edge $v_i \rightarrow v_j$, the cost of the $L_{i, j}$ will be counted multiple times. 

       Since this paper is focused on network design problem, and the link $L_{i, j}$ may be used to carry multiple flows, in the cost calculation phase, the cost of the link $L_{i,j}$ should be counted only one time instead of multiple times. For this reason objective function stated in equation (8) is not the appropriate objective function for the RDCMCCC problem. In the following subsection we discuss how additional constraints have to introduced so that cost calculation is carried out correctly in accordance with the stated objective of the RDCMCCC problem.

       The variation of the RDCMCCC problem given by the objective function (8), and constraints (4), (5), (6) and (7) is more in alignment with the versions of the Vehicle Routing Problem \cite{gillett1974heuristic}, as it makes prefect sense in that environment. If multiple vehicles uses the same road segment $v_i \rightarrow v_j$, the cost of this segment must be counted multiple times as some amount of gas will be consumed by each vehicle independently of other vehicles, and as such cost of traversing road segment $v_i \rightarrow v_j$ has to be counted multiple times. Similarly, if one vehicle uses the segment $v_i \rightarrow v_j$, and another vehicle uses $v_j \rightarrow v_i$, the cost of this segment must be counted two times.

       In Section \ref{CC_RDCMCCC} we showed that the RDCMCCC problem is NP-complete. An interesting feature of the variation of the RDCMCCC problem given by the objective function (8) and constraints (4), (5), (6) and (7) can be solved in polynomial time. In the following we show that the matrix formed by the constraints (4), (5), (6) and (7) is Totally Unimodular \cite{schrijver1998theory,ahuja1988network} and hence relaxing the integrality requirement of the variable $f^{k, r}_{i,j}$, a polynomial time algorithm can be developed using the Ellipsoid method \cite{schrijver1998theory}. In the following we show that the matrix formed by the constraints (4), (5), (6) and (7) is Totally Unimodular.


\begin{figure*}[htbp!]
    \centering
    \includegraphics[width = 1\linewidth]{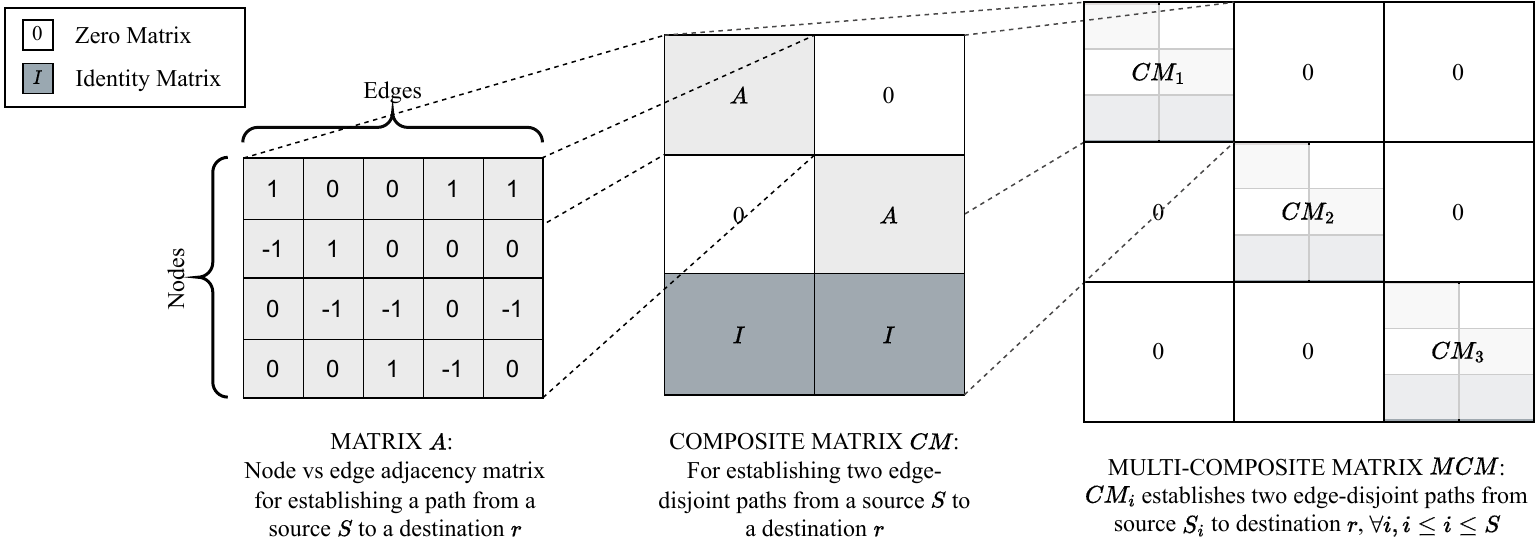}
    \vspace{-6mm}
    \caption{Constraint matrix of the ILP formulation of the (RDCMC-(R=1)-FTNDP) (MCM) and its formation from CM and Matrix A}   
    \label{fig_tum}  
    \vspace{-4mm}
\end{figure*}

The node-vs-(directed)-edge matrix of a graph with four nodes and 6 edges is shown in Fig. 2 (Matrix A). If only one path from a substation $v_i$ to the root node $v_1$ had to be established, this would have been the constraint matrix of the corresponding ILP. It is well known that this matrix is TUM \cite{ahuja1988network}. If two disjoint paths from from one node $v_i$ to $v_1$ had to be established, the constraint matrix would have looked like the Composite Matrix CM shown in Fig 2. In the following, we show that if Matrix A is a TUM, so is CM.

$$CM = \begin{pmatrix}
    A & 0 \\
    0 & A \\
    I & I
\end{pmatrix}$$
such that $A$ is a TUM where every column consists of a single $1$ and a single $-1$, and $I$ is the identity matrix.

\begin{theorem}
    $CM$ is a TUM.
\end{theorem}

Subtract the columns of the leftmost $I$ from the columns of the rightmost $I$ to get $\begin{pmatrix} A & -A \\ 0 & A \\ I & 0 \end{pmatrix}$ then add rows of the bottom $A$ to give $C' = \begin{pmatrix} A & 0 \\ 0 & A \\ I & 0 \end{pmatrix}$. These operations only change the sign of the determinant so it suffices to show that $C'$ is a TUM.

Consider a square submatrix of $C'$, call it $M$. If $M$ belongs to either $A$ component, either $0$ component, or the $I$ component it is obviously unimodular. If $M$ belongs to one of the components of the form $\begin{pmatrix}A \\ 0\end{pmatrix}$, $\begin{pmatrix}0 \\ A\end{pmatrix}$, $\begin{pmatrix}0 & A\end{pmatrix}$, $\begin{pmatrix}A & 0\end{pmatrix}$, $\begin{pmatrix}I \\ 0\end{pmatrix}$, or $\begin{pmatrix}0 \\ I\end{pmatrix}$ the determinant is $0$. If $M$ spans $\begin{pmatrix} A & 0 \\ 0 & A\end{pmatrix}$ every column, as stated above, has at most $2$ non-zero entries.

To see that $M$ if unimodular if it spans $\begin{pmatrix} 0 & A \\ I & 0\end{pmatrix}$ or $\begin{pmatrix} A & 0 \\ 0 & A \\ I & 0\end{pmatrix}$, induct on the size of $M$. No $1 \times 1$ matrix spans all these matrices, hence the base case is vascuously true. Now suppose it holds that every submatrix with dimension $n \times n$ has a determinant that is either $+1$, $-1$, or $0$, and suppose $M$ has dimension $n+1 \times n+1$. Then if $M$ does not span $\begin{pmatrix} 0 & A \\ I & 0\end{pmatrix}$ or $\begin{pmatrix} A & 0 \\ 0 & A \\ I & 0\end{pmatrix}$ it has already been shown to be a TUM via the argument above, so suppose $M$ spans $\begin{pmatrix} 0 & A \\ I & 0\end{pmatrix}$ or $\begin{pmatrix} A & 0 \\ 0 & A \\ I & 0\end{pmatrix}$. One can always find a row beginning with a $1$ and everything else is $0$ or everything is $0$. Then, cofactor expansion can reduce the determinant to that of an $n \times n$ matrix with the sign possibly switched. The inductive hypothesis has a determinant of either $0$, $1$, or $-1$, and so $M$ is unimodular. Furthermore, $C'$ is a TUM.

\begin{theorem}
    If the Composite Matrix $CM$ is a TUM so is the Multi-Composite Matrix MCM (shown in Fig. 2).
\end{theorem}

\begin{proof}
    See \cite{schrijver1998theory} (page 280).
\end{proof}

\subsection{Additional Constraints for the RDCMCCC Problem}

 \noindent
       The {\em Total flow} on the edge $e_{i, j}$ is denoted by ${TF}_{i,j}$ and is given as the sum of all the flows on that edge,
       \begin{equation}
           {TF}_{i,j} = \sum_{ k = 1}^s \sum_{ r = 1}^ {R + 1} f^{k, r}_{i,j}
       \end{equation}
           Since in our model, data flow is {\em unidirectional}, i.e., from the substations with PMUs, i.e., $u_i. 1 \leq i \leq s$ nodes to the Control Center, i.e., $v_0$, we treat flow on the edge $e_{i, j}$ (i.e., $i \rightarrow j$) to be different from the flow on the edge $e_{j, i}$ (i.e., $j \rightarrow i$). Accordingly, ${TF}_{i,j}$ may be quite different from ${TF}_{j,i}$. If ${TF}_{i,j} > 0$ or ${TF}_{j,i} > 0$ on a communication link, $L_{i, j}$ (may be an optical fiber cable) needs to be installed to connect two substations $v_i$ and $v_j$. We assume that the link $L_{i, j}$ is {bi-directional} and is able to carry data traffic in both $i \rightarrow j$ and $j \rightarrow i$ directions. We denote the {\em cost} of the link $L_{i, j}$ by $C_{i,j}$.

           
\vspace{0.1in}
\noindent
           We use a binary variable $X_{i,j}$ to indicate whether the link $L_{i, j}$ is installed or not. If the link $L_{i, j}$ is installed, we have to account for its cost $C_{i, j}$; otherwise, not.  Since the objective of the RDCMCC problem is to minimize the network design cost, it can be expressed with the following objective function:
           \begin{equation}
               {\tt Objective~Function:~Minimize} \sum_{i = 1}^n \sum_{j > i}^n C_{i, j}X_{i, j}
           \end{equation}
           
            \vspace{0.1in}
           \noindent
           The link $L_{i, j}$ must be  installed only if ${TF}_{i,j} \geq 1$ or ${TF}_{j,i} \geq 1$. We introduce a variable $MaxTF$ for this purpose and define it as follows 
           
        \begin{equation}
            {MaxTF}_{i,j} = max ({TF}_{i,j}, {TF}_{j,i}) 
        \end{equation}

           $MaxTF$ can be expressed with the following inequalities:
           \begin{equation}
               {MaxTF}_{i,j} \geq {TF}_{i,j}
               \end{equation}
               \begin{equation}
            {MaxTF}_{i,j} \geq {TF}_{j,i}
           \end{equation}

           
           
        \noindent
        The constraint on the binary variable $X_{i,j}$ can now be expressed as 
        
           \begin{equation}
                X_{i,j} = \left\{ \begin{array}{ll}
               1, & \mbox{if  $Max{TF} \geq 1$}\\
               0, & \mbox{otherwise}
                   \end{array}
           \right. 
           \end{equation}
           

           \vspace{0.1in}
           \noindent
           It may be noted that while in the objective function, the binary variable $X_{i, j}$ appears but the flow (path) variables, such as $f^{k, r}_{i,j}$, used to establish the paths from $u_i, 1 \leq i \leq s$ to $v_0$ do not appear. Since the objective function tries to minimize \(\sum_{i = 1}^n \sum_{j > i}^n C_{i, j}X_{i, j}\), unless some other constraints are specified, there is nothing to prevent all $X_{i,j}$ values being zero with minimum value of the objective function being equal to zero. The constraints that should prevent all $X_{i,j}$ values from being zero are the establishment of $s \times (R +1)$ paths from the set of nodes $U$ to the node $v_0$. This constraint can be ensured by the following inequality
        \begin{equation}
                  ~{MaxTF}_{i,j} \leq s \cdot X_{i,j} 
        \end{equation}

 
 In the RDCMCCC substation data must reach the Control Center within the specified delay threshold $\delta$. This implies that the delay on both the disjoint paths from the substation to the Control Center be less than $\delta$. This constraint can be imposed in the following way:
 

 \begin{equation}
      \forall k, 1\leq k \leq s, \forall r, 1 \leq r \leq R+1,\sum_{i = 0}^n \sum_{j = 0}^n D_{i, j}f^{k, r}_{i, j}\leq \delta
 \end{equation}

\section{Heuristic Solution for RDCMCCC Problem}

        \begin{table*}[h!]
\begin{adjustbox}{width = \textwidth}
\centering
\begin{tabular}{|c|c|c|c|c|c|c|c|c|c|c|c|c|} 
 \hline
{\bf Phoenix} & \multicolumn{12}{c|}{Ratio = Sweep\_RDCMCCC($P$, $\delta$)/ OPT\_RDCMCCC($P$, $\delta$) 
  (R=1) } \\ \hline \hline
            
   & \multicolumn{3}{c|}{$\delta$ = 100} & \multicolumn{3}{c|} {$\delta$ = 150} & \multicolumn{3}{c|} {$\delta$ = 200} & \multicolumn{3}{c|} {$\delta$ = 250}\\ \hline
  Data Sets&Sweep& Optimal&Ratio&Sweep& Optimal&Ratio&Sweep& Optimal&Ratio&
            Sweep& Optimal&Ratio \\ \hline
     \multirow{2}{*}{\centering $DS_1(35\%)$} & \begin{tabular}[t]{c} 774.29 \\ 0.00 \end{tabular} & \begin{tabular}[t]{c} 225.23 \\ 0.00 \end{tabular} & \begin{tabular}[t]{c} 3.44 \\ 0.00 \end{tabular} & \begin{tabular}[t]{c} 578.11 \\ 0.00 \end{tabular} & \begin{tabular}[t]{c} 216.44 \\ 0.00 \end{tabular} & \begin{tabular}[t]{c} 2.67 \\ 0.00 \end{tabular} & \begin{tabular}[t]{c} 531.79 \\ 0.00 \end{tabular} & \begin{tabular}[t]{c} 216.44 \\ 0.00 \end{tabular} & \begin{tabular}[t]{c} 2.46 \\ 0.00 \end{tabular} & \begin{tabular}[t]{c} 531.12 \\ 0.00 \end{tabular} & \begin{tabular}[t]{c} 215.00 \\ 0.00 \end{tabular} & \begin{tabular}[t]{c} 2.47 \\ 0.00 \end{tabular}
     \\ \hline
     \multirow{2}{*}{\centering $DS_2(30\%)$} & \begin{tabular}[t]{c} 546.33 \\ 23.09 \end{tabular} & \begin{tabular}[t]{c} 206.17 \\ 25.43 \end{tabular} & \begin{tabular}[t]{c} 2.68 \\ 0.44 \end{tabular} & \begin{tabular}[t]{c} 485.24 \\ 44.30 \end{tabular} & \begin{tabular}[t]{c} 204.28 \\ 17.74 \end{tabular} & \begin{tabular}[t]{c} 2.37 \\ 0.02 \end{tabular} & \begin{tabular}[t]{c} 428.44 \\ 54.27 \end{tabular} & \begin{tabular}[t]{c} 204.21 \\ 17.70 \end{tabular} & \begin{tabular}[t]{c} 2.10 \\ 0.14 \end{tabular} & \begin{tabular}[t]{c} 403.38 \\ 36.50 \end{tabular} & \begin{tabular}[t]{c} 203.16 \\ 16.80 \end{tabular} & \begin{tabular}[t]{c} 1.98 \\ 0.02 \end{tabular}
     \\
     \hline
     \multirow{2}{*}{\centering $DS_3(25\%)$} & \begin{tabular}[t]{c} 450.46 \\ 73.21 \end{tabular} & \begin{tabular}[t]{c} 190.58 \\ 17.44 \end{tabular} & \begin{tabular}[t]{c} 2.36 \\ 0.25 \end{tabular} & \begin{tabular}[t]{c} 401.27 \\ 58.92 \end{tabular} & \begin{tabular}[t]{c} 184.99 \\ 16.08 \end{tabular} & \begin{tabular}[t]{c} 2.17 \\ 0.23 \end{tabular} & \begin{tabular}[t]{c} 355.43 \\ 26.03 \end{tabular} & \begin{tabular}[t]{c} 182.10 \\ 16.33 \end{tabular} & \begin{tabular}[t]{c} 1.95 \\ 0.04 \end{tabular} & \begin{tabular}[t]{c} 357.70 \\ 27.62 \end{tabular} & \begin{tabular}[t]{c} 182.10 \\ 16.33 \end{tabular} & \begin{tabular}[t]{c} 1.97 \\ 0.15 \end{tabular}
     \\ 
     \hline
        \multirow{2}{*}{\centering $DS_4(20\%)$} & \begin{tabular}[t]{c} 420.97 \\ 127.36 \end{tabular} & \begin{tabular}[t]{c} 185.92 \\ 33.78 \end{tabular} & \begin{tabular}[t]{c} 1.92 \\ 0.24 \end{tabular} & \begin{tabular}[t]{c} 349.82 \\ 94.01 \end{tabular} & \begin{tabular}[t]{c} 178.64 \\ 32.39 \end{tabular} & \begin{tabular}[t]{c} 1.95 \\ 0.31 \end{tabular} & \begin{tabular}[t]{c} 338.21 \\ 105.09 \end{tabular} & \begin{tabular}[t]{c} 174.58 \\ 29.06 \end{tabular} & \begin{tabular}[t]{c} 1.92 \\ 0.43 \end{tabular} & \begin{tabular}[t]{c} 310.87 \\ 94.48 \end{tabular} & \begin{tabular}[t]{c} 174.58 \\ 29.06 \end{tabular} & \begin{tabular}[t]{c}1.76 \\0.35  \end{tabular}
     \\ 
     \hline
      \multirow{2}{*}{\centering $DS_5(15\%)$} & \begin{tabular}[t]{c} 251.54 \\ 18.53 \end{tabular} & \begin{tabular}[t]{c} 144.18 \\ 3.92 \end{tabular} & \begin{tabular}[t]{c} 1.74 \\ 0.09 \end{tabular} & \begin{tabular}[t]{c} 226.60 \\ 37.15 \end{tabular} & \begin{tabular}[t]{c} 140.00 \\ 5.48 \end{tabular} & \begin{tabular}[t]{c} 1.61 \\ 0.21 \end{tabular} & \begin{tabular}[t]{c} 238.40 \\ 76.09 \end{tabular} & \begin{tabular}[t]{c} 140.00 \\ 5.48 \end{tabular} & \begin{tabular}[t]{c} 1.69 \\ 0.47 \end{tabular} & \begin{tabular}[t]{c} 211.01 \\ 37.60 \end{tabular} & \begin{tabular}[t]{c} 140.00 \\ 5.48 \end{tabular} & \begin{tabular}[t]{c} 1.50 \\ 0.21 \end{tabular}
\\
\hline
\end{tabular}
\end{adjustbox}
\vspace{2mm}
\caption{Ratio between the Sweep and Optimal solutions for Phoenix}
\label{table1}
\end{table*}

\begin{table*}[t]
\begin{adjustbox}{width = \textwidth}
\centering
\begin{tabular}{|c|c|c|c|c|c|c|c|c|c|c|c|c|} 
 \hline
{\bf Tucson} & \multicolumn{12}{c|}{Ratio = Sweep\_RDCMCCC($P$, $\delta$)/ OPT\_RDCMCCC($P$, $\delta$) (R=1)}\\ \hline \hline
            
   & \multicolumn{3}{c|}{$\delta$ = 100} & \multicolumn{3}{c|} {$\delta$ = 150} & \multicolumn{3}{c|} {$\delta$ = 200} & \multicolumn{3}{c|} {$\delta$ = 250}\\ \hline
  Data Sets&Sweep& Optimal&Ratio&Sweep& Optimal&Ratio&Sweep& Optimal&Ratio&
            Sweep& Optimal&Ratio \\ \hline
     \multirow{2}{*}{\centering $DS_1(60\%)$} & \begin{tabular}[t]{c} 202.38 \\ 54.64 \end{tabular} & \begin{tabular}[t]{c} 155.87 \\ 54.14 \end{tabular} & \begin{tabular}[t]{c} 1.32 \\ 0.14 \end{tabular} & \begin{tabular}[t]{c} 168.95 \\ 33.79 \end{tabular} & \begin{tabular}[t]{c} 151.14 \\ 51.79 \end{tabular} & \begin{tabular}[t]{c} 1.26 \\ 0.08 \end{tabular} & \begin{tabular}[t]{c} 182.25 \\ 59.71 \end{tabular} & \begin{tabular}[t]{c} 149.66 \\ 53.62 \end{tabular} & \begin{tabular}[t]{c} 1.22 \\ 0.04 \end{tabular} & \begin{tabular}[t]{c} 181.26 \\ 58.09 \end{tabular} & \begin{tabular}[t]{c} 148.34 \\ 51.44 \end{tabular} & \begin{tabular}[t]{c} 1.22 \\ 0.03 \end{tabular}
     \\ \hline
     \multirow{2}{*}{\centering $DS_2(55\%)$} & \begin{tabular}[t]{c} 221.31 \\ 0.00 \end{tabular} & \begin{tabular}[t]{c} 161.51 \\ 0.00 \end{tabular} & \begin{tabular}[t]{c} 1.37 \\ 0.00 \end{tabular} & \begin{tabular}[t]{c} 222.09 \\ 33.14 \end{tabular} & \begin{tabular}[t]{c} 182.16 \\ 45.22 \end{tabular} & \begin{tabular}[t]{c} 1.25 \\ 0.15 \end{tabular} & \begin{tabular}[t]{c} 214.93 \\ 32.99 \end{tabular} & \begin{tabular}[t]{c} 180.56 \\ 45.82 \end{tabular} & \begin{tabular}[t]{c} 1.21 \\ 0.12 \end{tabular} & \begin{tabular}[t]{c} 213.52 \\ 30.60 \end{tabular} & \begin{tabular}[t]{c} 178.97 \\ 43.07 \end{tabular} & \begin{tabular}[t]{c} 1.21 \\ 0.12 \end{tabular}
     \\
     \hline
     \multirow{2}{*}{\centering $DS_3(50\%)$} & \begin{tabular}[t]{c} 149.43 \\ 0.01 \end{tabular} & \begin{tabular}[t]{c} 115.84 \\ 0.00 \end{tabular} & \begin{tabular}[t]{c} 1.29 \\ 0.00 \end{tabular} & \begin{tabular}[t]{c} 153.88 \\ 43.84 \end{tabular} & \begin{tabular}[t]{c} 133.58 \\ 46.02 \end{tabular} & \begin{tabular}[t]{c} 1.17 \\ 0.06 \end{tabular} & \begin{tabular}[t]{c} 160.35 \\ 55.05 \end{tabular} & \begin{tabular}[t]{c} 131.40 \\ 42.24 \end{tabular} & \begin{tabular}[t]{c} 1.21 \\ 0.02 \end{tabular} & \begin{tabular}[t]{c} 153.21 \\ 42.68 \end{tabular} & \begin{tabular}[t]{c} 131.40 \\ 42.24 \end{tabular} & \begin{tabular}[t]{c} 1.17 \\ 0.05 \end{tabular}
     \\ 
     \hline
        \multirow{2}{*}{\centering $DS_4(45\%)$} & \begin{tabular}[t]{c} 165.31 \\ 22.70 \end{tabular} & \begin{tabular}[t]{c} 137.86 \\ 14.48 \end{tabular} & \begin{tabular}[t]{c} 1.20 \\ 0.08 \end{tabular} & \begin{tabular}[t]{c} 157.68 \\ 24.36 \end{tabular} & \begin{tabular}[t]{c} 130.43 \\ 14.11 \end{tabular} & \begin{tabular}[t]{c} 1.21 \\ 0.09 \end{tabular} & \begin{tabular}[t]{c} 156.69 \\ 23.00 \end{tabular} & \begin{tabular}[t]{c} 130.43 \\ 14.11 \end{tabular} & \begin{tabular}[t]{c} 1.20 \\ 0.08 \end{tabular} & \begin{tabular}[t]{c} 156.69 \\ 23.00 \end{tabular} & \begin{tabular}[t]{c} 130.43 \\ 14.11 \end{tabular} & \begin{tabular}[t]{c} 1.20 \\ 0.08 \end{tabular}
     \\ 
     \hline
      \multirow{2}{*}{\centering $DS_5(40\%)$} & \begin{tabular}[t]{c} 199.25 \\ 53.79 \end{tabular} & \begin{tabular}[t]{c} 143.01 \\ 3.83 \end{tabular} & \begin{tabular}[t]{c} 1.39 \\ 0.34 \end{tabular} & \begin{tabular}[t]{c} 206.21 \\ 41.69 \end{tabular} & \begin{tabular}[t]{c} 157.72 \\ 40.36 \end{tabular} & \begin{tabular}[t]{c} 1.32 \\ 0.07 \end{tabular} & \begin{tabular}[t]{c} 198.87 \\ 50.27 \end{tabular} & \begin{tabular}[t]{c} 156.04 \\ 37.45 \end{tabular} & \begin{tabular}[t]{c} 1.27 \\ 0.10 \end{tabular} & \begin{tabular}[t]{c} 190.83 \\ 37.41 \end{tabular} & \begin{tabular}[t]{c} 156.04 \\ 37.45 \end{tabular} & \begin{tabular}[t]{c} 1.23 \\ 0.12 \end{tabular}
\\
\hline
\end{tabular}
\end{adjustbox}
\vspace{2mm}
\caption{Ratio between the Sweep and Optimal solutions for Tucson}
\label{table2}
\end{table*}

In this section we provide a heuristic solution for the RDCMCCC problem. The locations of the substations and the Control Center are considered as $n$ points on a two dimensional plane. The rectangular (i.e., $(x, y)$) coordinates of the Control Center is taken to be (0, 0). The location of the $i$-th substation (i.e., point $p_i$) is expressed in polar coordinate as $(r_i, \theta_i)$, where $r_i$ denotes its distance from the Control Center, and $\theta_i$ denotes its polar coordinate angle. The substation locations points, $P = \{p_1, p_2, p_3, \ldots, p_n\}$ are numbered in non-decreasing order of their polar angle, i.e., $\forall i, 1 \leq i \leq {n - 1}, \theta_i \leq \theta_{i + 1}$.
The heuristic is based on the idea of a vector first connecting the point (0,0) to point $(r_1, \theta_1)$ and then rotating in an anti-clockwise fashion sequentially sweeping the rest of the set of points \{$(r_2, \theta_2), (r_3, \theta_3),, \ldots, (r_n, \theta_n)\}$. In other words, the vector {\em sweeps} the points one after another in non-decreasing order of their polar angle. Because of the nature of its working, we refer to it as the {\em SWEEP} heuristic. It attempts to construct the first  cycle $C_1$, starting and ending at (0, 0), and including the largest number of points \{$(r_1, \theta_1), (r_2, \theta_2), \ldots, (r_k, \theta_a)\}$, as long as the length of the cycle with points $\{(0, 0), (r_1, \theta_1), (r_2, \theta_2), \ldots, (r_k, \theta_a)\}$ does not exceed the acceptable delay threshold $\delta$. It then repeats the process by trying to create a second cycle $C_2$ with points $\{(0, 0), (r_{a+1}, \theta_{a+1}), (r_{a+2}, \theta_{a+2}), \ldots, (r_b, \theta_b)\}$ without exceeding the delay threshold $\delta$. This repetition of cycle creation process continues till all the points are included in at least one cycle.

\begin{algorithm}
\SetAlgoNlRelativeSize{0}
\SetNlSty{textbf}{(}{)}
\SetAlgoNlRelativeSize{-1}

\KwIn{A set of points on a plane: $P=\{p_0,p_1,p_2,\cdots,p_n\}$, 
 $p_0$ is the root node( i.e the location of control center), Maximum cycle length: $\delta$}
\KwOut{A set of cycles $\mathcal{C}$}

\caption{Sweep Heuristic for RDCMCCC}

\BlankLine

Sort the points in $P$ based on their polar angle with respect to the root node $p_0$.

$\mathcal{C} := \varnothing$

Initialize  $j$=1

$C_{1} := \{p_0\}$.


\tcc{Creation of First cycle}
\While{$j \leq n$}{
\If{length$(C_1 \cup \{p_j\}) \leq \delta$}{

$C_{1} \leftarrow C_{1} \cup \{p_j\}$.

$j \leftarrow j+1$
}

\Else{

$\mathcal{C} \leftarrow \mathcal{C} \cup \{C_{1}\}$


break

}

}

Initialize $i$=2 

\tcc{Creation of subsequent cycles}

\While{$j \leq n$}{

Create new cycle $C_{i}$ after analysing all the candidates created with the help of cycle $C_{i-1}$(as discussed in Section V).

Set $j \leftarrow j+x$ ,where $x$ is the number of new points included in the cycle $C_{i}$

$\mathcal{C} \leftarrow \mathcal{C} \cup  \{C_{i}\}$

$i \leftarrow i+1$

}

Return $\mathcal{C}$ which contains set of cycles that covers all the points

\end{algorithm}

During the creation of cycles $C_2$ and beyond, some refinements are done to reduce the total design cost of the cycles. The design cost is the sum of the cost of the edges of that make up the cycles. In the following, we discuss the refinement scheme in detail.\\
\noindent
$\cal C$: The set of cycles $\{C_1, C_2, \ldots \}$ created by the Sweep heuristic\\
$P$: The set of points $\{p_0, p_1, p_2, \ldots p_n\}$ representing the locations of the Control Center and the substations. The point $p_0$ represent the location of the CC and the rest of the points represent the locations of the substations

\vspace{0.05in}
\noindent
Suppose that the first cycle $C_1$ created by SWEEP includes the points $p_0, p_1, p_2, \ldots, p_a$, i.e.,\\
$C_1: p_0 \rightarrow p_1 \rightarrow p_2 \rightarrow \ldots \rightarrow p_a \rightarrow p_0$

\vspace{0.05in}
\noindent
Suppose that at some stage of execution of SWEEP, the cycles $C_1, C_2, \ldots, C_{i}$  have been created and it's about to create the cycle $C_{i + 1}$. 

\vspace{0.05in}
\noindent
Suppose that the points on the $i$-th cycle $C_i$ are denoted by $p_0, q_i^1, q_i^2, \ldots, q_i^{a_i}$, where  $q_i^1, q_i^2, \ldots, q_i^{a_i}$ correspond to the points $p_r, p_{r +1}, \ldots, p_{r + {a_i} - 1}$  of the set of points $P$.

\vspace{0.05in}
\noindent
The cycle $C_{i+1}$ is determined by first creating a set of $a_i$ potential candidates of $C_{i+1}$.  We denotes these candidates as $C^1_{i+1}, C^2_{i+1}, \ldots, C^{a_i}_{i+1}$. Each candidate $C^s_{i+1}, 1 \leq s \leq a_i$ includes some of the points of the immediately preceding cycle $C_i$ and some new points $p_b, p_{b+1}, \ldots, p_{b + c}$ which were not part of any previously generated cycles $C_1$ through $C_i$.

\vspace{0.05in}
\noindent
If the $i$-th cycle $C_i$ comprises of $p_0, q_i^1, q_i^2, \ldots, q_i^{a_i}$, the $w$-th candidate of the $i + 1$-th cycle (i.e., $C^w_{i+1}, 1 \leq w \leq a_i$) will comprise of either $p_0, q_i^1, q_i^2, \ldots, q_i^{w}, p_b, p_{b+1}, \ldots, p_{b + c}$ or $p_0, q_i^{a_i}, q_i^{a_i - 1}, \ldots, q_i^{w}, p_b, p_{b+1}, \ldots, p_{b + d}$, depending on whether the length of the path $p_0, q_i^1, q_i^2, \ldots, q_i^{w}$  is shorter or longer than the path $p_0, q_i^{a_i}, q_i^{a_i - 1}, \ldots, q_i^{w}$ on the cycle $C_i$.

\vspace{0.05in}
\noindent
Each candidate for $C_{i +1}$  will have zero or more points from the previous cycle $C_i$ and will include some new points which were not part of any previous cycle. The reason for using some of the points from the previous cycle $C_i$ is that the link cost of connecting these points have already been paid for during the construction of $C_i$ and as such no additional cost for these links are incurred during the construction of $C_{i + 1}$. Furthermore, each candidate solution of $C_i$ includes some new points and also has a certain cost (this is the cost of the new links to be introduced, excluding the ones that were already paid for during the construction of the previous cycle). If a candidate solution includes $x$ new points with $y$ amount of additional link cost, the candidate that has the best {\em cost/benefit} ratio is chosen as the cycle $C_{i +1}$. It may be noted that {\em benefit} in this context is the number of new nodes included and the {\em cost} is the new link cost needed for creation of this candidate solution.

\section{Experimental Results}

The results of our experimental evaluation of the \textit{SWEEP} heuristic with substation location data of Arizona are presented here. The latitude-longitude locations of substations in Arizona were collected from the U.S. Dept. of Homeland Security website. 
The total number of substations in Arizona is 892, of which, 653 are operational. As PMUs are expensive, not all substations have them. It has been reported in \cite{baldwin1993power} 
that PMUs installed in 20\%-30\% of the substations are sufficient for full {\em observability}. 
We conduct our experiments with substation locations in Phoenix and Tucson, the two largest cities in Arizona. The number of substations in Phoenix and Tucson are 132 and 25, respectively. The results of the experiments (where R=1) are presented in Tables \ref{table1} and \ref{table2}. The data sets $DS_1$ through $DS_5$ correspond to differing percentages of the substations with PMUs. In Table \ref{table1}, $35\%$ next to $DS_1$ indicates that in that data set, $35\%$ of 132 substations (46) are assumed to have PMUs. The delay threshold value $\delta$ was varied from 100 to 250. 46 substation locations from 132 were selected randomly three times. These three data sets may be viewed as $DS_{1,1}, DS_{1,2}, DS_{1.3}$. The design costs with SWEEP, Optimal, and the ratio between the two 
were computed for these three data sets for a specific $\delta$ value, and the corresponding {\em average} and {\em standard deviation} are presented in Table \ref{table1}. The values for SWEEP, Optimal, Ratio for $DS_5(15\%)$ and $\delta = 100$ are (251.54, 18.53), (144.18, 3.92) and (1.74, 0.09) respectively. The other entries in tables \ref{table1} and \ref{table2} were similarly computed. As the number of substations in Tucson was fewer, we varied the percentage from $40\%$ to $60\%$, instead of $15\%$ to $35\%$ as was done for Phoenix. 
The ratios of the SWEEP to Optimal range from 1.5 to 3.44 for Phoenix dataset and from 1.2 to 1.39 for the Tucson dataset. The computational times for the optimal solution was multiple hours, whereas the SWEEP only took a few seconds. Accordingly, one can conclude that SWEEP is quite effective. Please note that some of the standard deviation values in the tables are zeroes. This is due to the fact that for some datasets, SWEEP produced degenerate solutions and the Optimal found only one solution after 24 hours of computation.

\section{Conclusion}

We have investigated a generalization of minimum delay constrained spanning trees: instead of considering trees, we considered subgraphs where removal of $R$ edges does \textit{not} disconnect any vertex from a designated root node. In particular, we focused on the special case where $R = 1$, which we refer to as the RDCMCCC problem, that is, removal of a single edge does not disconnect the edges from the source node while still satisfying the delay constraints. We have shown that even when $R = 1$ the problem is NP-complete. Despite the difficulty of the problem, we have presented a novel ILP formulation based on network flows, and shown that in some special cases, the resulting constraint matrix is likely unimodular. We have also created a polynomial time heuristic which has been used to solve the Euclidean version of RDCMCCC. 

In the future we hope to extend this work by considering other heuristics that may perform better. In addition, we hope that there may be fast approximation algorithms that grant sharp theoretical bounds on the cost of such structures.

\bibliographystyle{IEEEtran}
\bibliography{DRCN_Ref2}
\end{document}